\documentclass[a4paper,UKenglish]{lipics-v2018}

\usepackage{tikz}
\tikzset{every picture/.style={/utils/exec={\sffamily}}}
\usepackage{siunitx}
\usepackage[noend]{algpseudocode}
\graphicspath{/pic}
\usetikzlibrary{intersections}
\tikzset{every picture/.style={line width=0.8pt}}
\usepackage{amsmath} 
\usepackage{amsfonts} 
\usepackage{amssymb} 
\usepackage{xspace}   
\usepackage{relsize} 
\usepackage{verbatim} 
\usepackage{enumerate} 
\usepackage{enumitem} 
\usepackage{centernot} 
\usepackage{adjustbox} 
\usepackage{multirow} 
\usepackage{stmaryrd} 
\usepackage{units} 
\usepackage[normalem]{ulem} 
\usepackage{ifthen} 
\usepackage{stmaryrd}

\newcommand{\margincomment}[1]%
{{\marginpar{{\footnotesize\begin{minipage}{0.75in}%
          \begin{flushleft}%
            {#1}%
          \end{flushleft}%
        \end{minipage}%
      }}}\ignorespaces}

\newcommand{\mareksmargincomment}[1]{}
\newcommand{\mareksintextcomment}[1]{}

\colorlet{darkgreen}{green!45!black}

\newcommand{\myparagraph}[1]{{\smallskip\par\noindent{\bf #1}}}
\newcommand{\etal}{\emph{et al}.\xspace}

\newcommand{\mycase}[1]{\mbox{{\underline{Case #1}}:\/}}

\newcommand{\tildeG}{{\tilde{G}}}
\newcommand{\tildes}{{\tilde{s}}}
\newcommand{\tildet}{{\tilde{t}}}

\newcommand{\dualprec}{{\,\trianglelefteq\,}}

\newcommand{\aleft}{a_{\scriptscriptstyle\textit{L}}}
\newcommand{\dleft}{d_{\scriptscriptstyle\textit{L}}} 
\newcommand{\aright}{a_{\scriptscriptstyle\textit{R}}}  
\newcommand{\dright}{d_{\scriptscriptstyle\textit{R}}}   

\newcommand{\aleftleft}{{a}^1_{\scriptscriptstyle\textit{L}}}
\newcommand{\dleftleft}{{d}^1_{\scriptscriptstyle\textit{L}}} 
\newcommand{\arightright}{{a}^2_{\scriptscriptstyle\textit{R}}}  
\newcommand{\drightright}{{d}^2_{\scriptscriptstyle\textit{R}}}  

\newcommand{\TPL}{${\textrm{SP}}$}

\nolinenumbers

\title{Modeling Fluid Mixing in Microfluidic Grids}

\titlerunning{Fluid Mixing in Grids}

\author{Huong Luu}{Department of Computer Science\\ University of California at Riverside}{}{}{}

\author{Marek Chrobak}{Department of Computer Science\\ University of California at Riverside}{}{}{}

\Copyright{The copyright is retained by the authors}

\authorrunning{H.~Luu and M.~Chrobak}

\subjclass{%
}


\funding{Research supported by NSF grant CCF-1536026.}

\begin{document}

\maketitle

\begin{abstract}
We describe an approach for modeling fluid concentration profiles in grid-based 
microfluidic chips for fluid mixing. This approach provides an 
algorithm that predicts fluid concentrations at the chip
outlets. Our algorithm significantly outperforms COMSOL finite
element simulations in term of runtime while still produces
results that closely approximate those of COMSOL.
\end{abstract}


\section{Introduction}
\label{sec: introduction}


Microfluidics is an emerging technology for manipulating nanoliter-scale fluid volumes,
with applications in a variety of fields including biology, chemistry, biomedicine, and materials science.
Fast progress in this area led to the development of microfluidic chips (MFCs), which are integrated microfluidic 
devices that are increasingly often used in various laboratory processes such as 
medical diagnosis~\cite{yager_etal_microfluidic_2006}, DNA purification~\cite{kastania_plasma_2016}, 
or cell lysis~\cite{jungkyu_microfluidic_2009}. 
MFCs offer a solution to automate laboratory experiments that saves time, reduces labor costs, limits usage of
chemical reagents, and replaces complex and expensive equipments~\cite{squires_microfluidic_2005, jungkyu_microfluidic_2009}. 

One function often implemented on MFCs is fluid mixing. The technologies for
mixing MFCs currently in use fall into two broad categories:
droplet-based chips, where the fluid is manipulated in discrete units called droplets,
and flow-based chips, based on continuous flow. In this work we focus exclusively on the flow-based model.

Fluid mixing is particularly important in sample preparation, where the objective is to dilute the sample fluid, 
also called \emph{reactant}, using another fluid that we refer to as \emph{buffer}. 
For example, in cell lysis, the sample preparation process includes a step of mixing blood sample 
with citrate buffer~\cite{jungkyu_microfluidic_2009}.
For some experimental processes samples with multiple pre-specified volumes and concentrations are needed.
For instance, such a sample may 
consist of $5\mu L$ of reactant with concentration $10\%$, $10\mu L$ of reactant with concentration $20\%$,
and $10\mu L$ of reactant with concentration $40\%$. 
Samples involving such multiple target concentrations are common in preclinical drug development processes,
and they are also used for other experiments, for example in biochemical assays~\cite{bhattacharjee_algorithm_2019}. 
For these applications, one needs to design a MFC that produces the specified target set of 
concentration/volume pairs of reactant.  

Several flow-based designs have been proposed in the literature.
In~\cite{paegel_microfluidic_2006}, the authors proposed an MFC that uses two-way valves to produce serial dilution. 
An electrokinetically driven MFC design was introduced in~\cite{stephen_microfluidic_1999} for serial mixing. 
The above approaches require different designs by changing valves or splitter channels placement,
or tuning voltage control to create different target sets. 
In~\cite{bhattacharjee_dilution_2017}, the authors gave a dilution algorithm for given target concentration ratios using rotary mixers.
However, their method produces waste, and it also uses valves which can complicate the fabrication process. 


\myparagraph{Grid mixers.}
A very different approach was developed by Wang~{\etal}~\cite{wang_random_2016}.
Their proposed solution involves creating a library of ready-to-use micromixers that
users can query to find chip designs with desired properties. 
Their MFCs are simple rectilinear grids with two inputs (one for reactant and one for buffer)
and three outputs, thus capable of producing a set of three different concentrations.   
They do not require any valves nor any other functional elements.
A user identifies an appropriate design by submitting a query consisting of the desired reactant concentrations.
In their approach the design process is eliminated and the database is created
by exploring a large collection of randomly generated grids. For each random grid, its outlet concentration
values are computed by simulating fluid flow through the grid using COMSOL Multiphysics® software
(a commercial software that uses finite element analysis method to model physics processes, including fluid dynamics).
As shown in~\cite{wang_random_2016}, these COMSOL simulation results provide very accurate prediction of
outlet concentrations in actual fabricated MFCs.

Exploring such random designs is extremely time consuming, as most randomly generated designs are actually
not useful, either because they are redundant or because they produce concentrations that are of little interest
(say, only near-pure reactant or buffer). Thus to produce a desired number of designs for the grid library,
one may need to examine many orders of magnitude more random designs. Indeed, in our experiment, to 
produce a collection of  2600 sufficiently different concentration triplets we needed to generate
50 millions $12 \times 12$ random grid designs.

With the need to test so many designs,
this approach can only be used for small size grids because of the simulation bottleneck of COMSOL Multiphysics®.
The simulation time for each design in~\cite{wang_random_2016} is roughly a minute. It is also not scalable to bigger grids. 
COMSOL Multiphysics® takes about 6 minutes to run a simulation 
for a $12 \times 12$ grid with the same mesh setting.
The process can be sped up by using a coarser mesh, but this results in
degraded accuracy of concentration predictions at the outlets.
Further, some users may prefer to design custom grids for their choices of the attributes:
velocity, solute, outlets'  locations, diffusion coefficient of reactant, etc. Such users would
need to have access to an often costly 
computational fluid dynamics software in order to be
able to run the simulations. 

The approach based on random grid generation was also considered in~\cite{weiquing_more_2018}, where the
authors propose a method to remove redundant channels to make the design process more efficient, simplifying
fabrication of grid MFCs and reducing reactant usage.


\myparagraph{Our contribution.}
Addressing this performance bottleneck in populating the grid library in the approach from~\cite{wang_random_2016},
we developed an algorithm to model fluid mixing in microfluidic grids, in order to 
predict the reactant concentrations at the outlets. 
In our approach, concentration profiles in grid channels are approximated using a simple
3-piece linear function, which allows us to simulate the mixing process in time linear in the
grid size.
The overall algorithm is scalable, simple, and 
produces good approximation of concentration values and flow rates   
at outlets in grid-based microfluidic chips, as compared to the results from COMSOL Multiphysics®. 
It is also much more general than the model from~\cite{wang_random_2016,weiquing_more_2018}, as
it allows grids of all sizes, any number of outlets, arbitrary inflow velocities,
and arbitrary fluids.
We also developed a web-based implementation of our algorithm that allows users to customize their
designs manually to match their needs\footnote{A prototype of this implementation is
available at \texttt{http://algorithms.cs.ucr.edu}.}. 

We add that our algorithm is not meant to be a complete substitute for COMSOL simulation. In the application to
populating mixing grid libraries, as in~\cite{wang_random_2016}, the recommended process would be to
apply our algorithm to select a collection of randomly generated designs with
potentially useful concentration vectors, and then re-verify these designs with COMSOL.
This way, COMSOL simulation will be performed only 
on a tiny fraction of generated grids --- less than $\%.01$, according
to our experiments. This significantly improves the overall performance.

While the main objective of this work was to develop an efficient numerical simulation algorithm, 
it also involves interesting combinatorial and topological aspects, as the correctness of our profile 
concentration model relies critically on duality properties of planar acyclic digraphs.


\section{Statement of the Problem}
\label{sec: statement of the problem}



\begin{figure}
	\hfill
	\includegraphics[height = 1.75in]{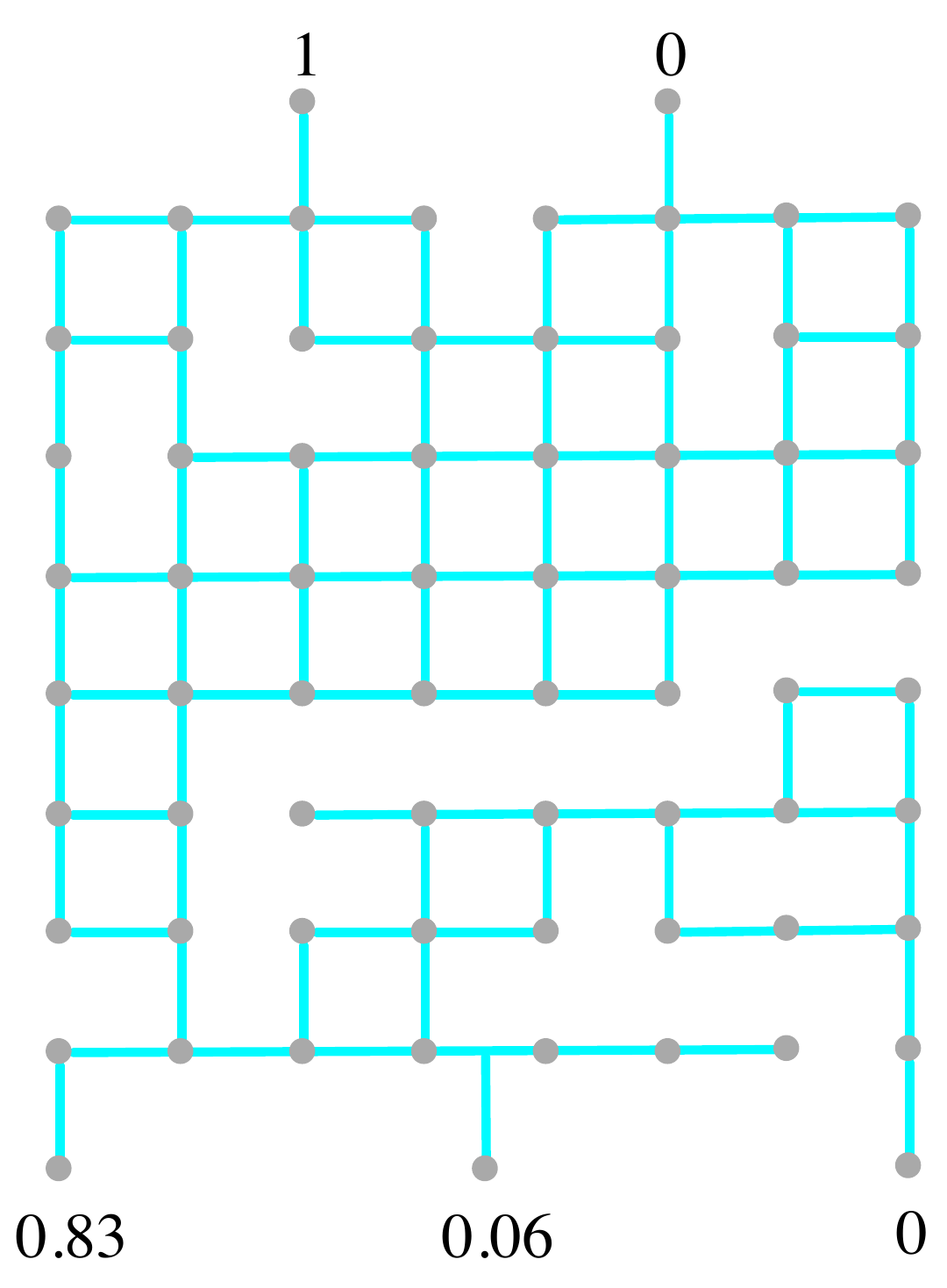}
	\hfill
	\includegraphics[height = 1.75in]{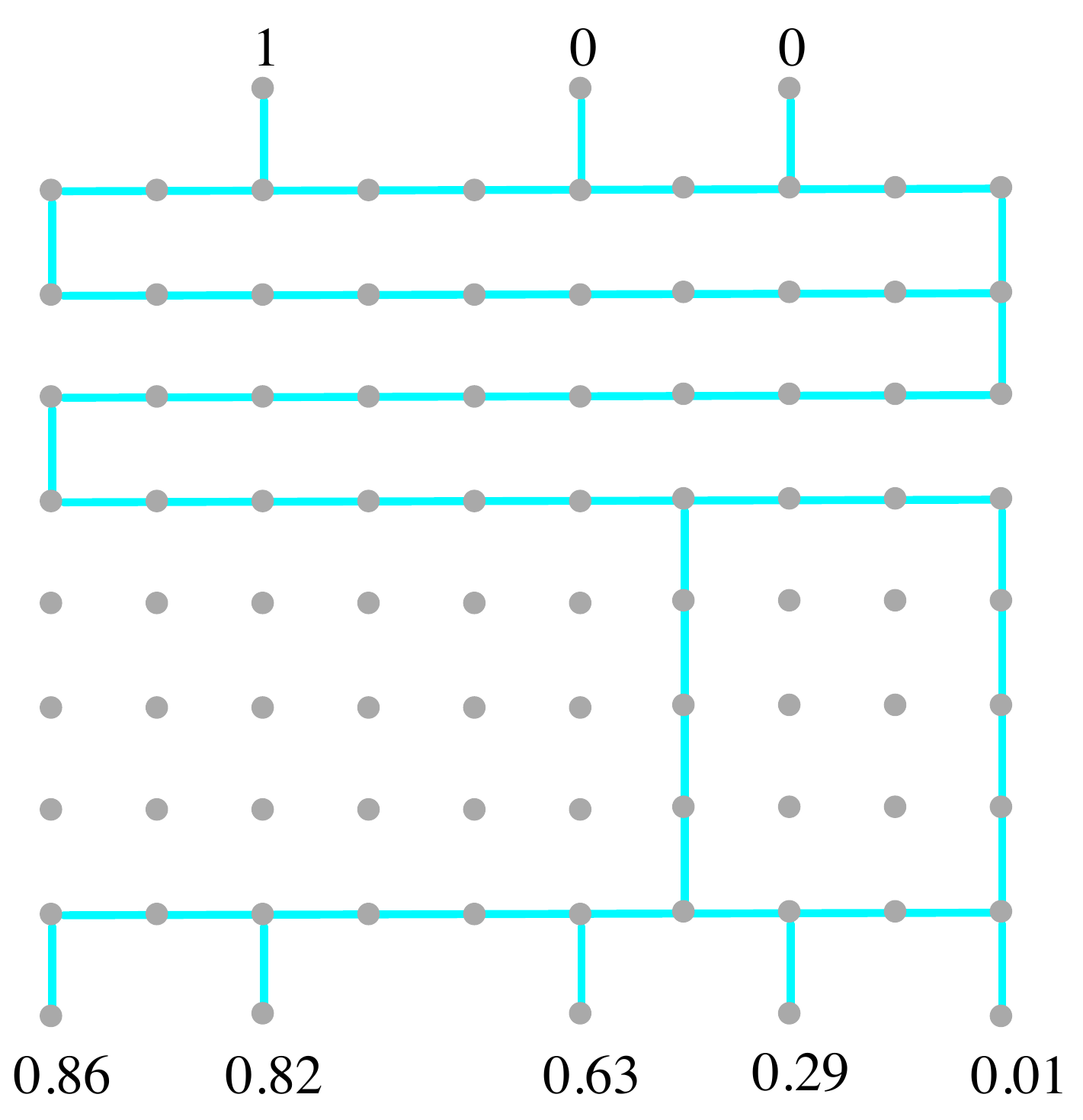}
	\hfill {\ }

	\caption{On the left, an example of an 8x8 grid from~\cite{wang_random_2016}.
	On the right, an example of a 10x8 grid in our model, with
	3 inlets and 5 outlets. The numbers at the outlets show the
	reactant concentrations.} 
	\label{fig_sample_mxn}			
\end{figure}			

We study grid-based MFCs for fluid mixing introduced in~\cite{wang_random_2016}.
Their model is $8\times 8$ grid with 2 inlets along the top edge of the grid and 3 outlets at the bottom, 
as shown in Figure~\ref{fig_sample_mxn}. 
The left inlet contains reactant, with concentration value 1, and the
right inlet contains buffer, with concentration value 0.
The channel width is 0.2 mm, and the channel length (distance between two grid vertices) is 1.5 mm.
The fluid velocity in the inlets is 10 mm/s. The outlets' pressure is 0 Pa.
The reactant is either sodium, fluorescein or bovine serum albumin.

We generalize the model from~\cite{wang_random_2016} in several ways. 
We allow arbitrary $m\times n$ grids, with any number of inlets and outlets (see Figure~\ref{fig_sample_mxn}).
The inlets are located along the top edge of the grid and the outlets at the bottom.
The inlets' solutions can take any concentration from 0 to 1, but they
must satisfy the following \emph{inlet monotonicity property}:
the inlet concentrations need to be non-increasing from left to right.
The inflows are of given constant rate, and the pressure values at all outlets are $0$.
Our model assumes that the flow throughout the grid is laminar, which is the case in standard
microfluidic applications.

In this setting, the problem we address can be formulated as follows: 
We are given  
(1) the specification of a grid design, and
(2) fluid properties, namely
its concentration and velocity at each inlet and its diffusion
coefficient. The goal is to determine the fluid concentration and velocity values at the outlets.


\section{Overview of the Algorithm}
\label{sec: overview of algorithm}


Our algorithm for predicting reactant concentrations at the outlets is based
on modeling its concentration profiles in grid's channels. 
Such a concentration profile
is a function that represents concentration values of the reactant along a line
perpendicular to the channel.
When fluid flows through straight segments of the grid,
this profile changes according to the laws of diffusion. In a node of the grid,
a flow may be split or several flows may be joined, 
and the profile changes
accordingly, producing complex non-linear functions.
The main idea behind our algorithm
is to approximate this profile using a simple 3-piece linear function. 
Once the profile at an output channel is computed, it determines the reactant concentration at this outlet.

We now give an overview of our algorithm, with more
detailed descriptions given in the sections that follow. 
\begin{description}

\item{(1)} Verify the correctness of the grid design, namely whether
	each edge (channel) is on at least one inlet-to-outlet path.
	(In our implementation, spurious fragments of the grid are automatically removed.)
	
\item{(2)} Compute the flow rates in each channel and pressure values at each node
	(see Section~\ref{sec: computing flow}).

\item{(3)} Partition the grid into parts, each part being either a 
	straight channel or a node. Depending on flow direction,
	a node can be one of three types:
	a join node (2-way or 3-way), a split node (2-way or 3-way), 
	or a combined join/split node (with 2 inflows and 2 outflows). 
	Sort these parts in an order consistent with the flow direction.
	(Once the flows are computed, one can think of the grid design as an
	acyclic directed graph. The desired order is then any topological sort
	of this acyclic graph.)

\item{(4)} Process the grid parts, in the earlier determined order, computing
	approximate concentration profiles (see Section~\ref{sec: algorithm for concentration profile}):
	\begin{itemize}
		\item For straight channels, the concentration profile at the
			end of the channel is determined from the profile at the
			beginning of the channel, based on the time the flow spends in this channel.
			This time is computed from the channel length and flow velocity.
		
		\item For nodes representing flow splits, split the
			incoming profile into outgoing profiles according to
			flow velocity ratios.
		
		\item For nodes representing flow joins, join the incoming
			profiles into the outgoing profile according to
			velocity ratios. This outgoing profile is then
			approximated by a 3-piece linear function.
	\end{itemize}
	
\item{(5)} Once all flow profiles in the grid are determined, for each outlet
	compute its fluid concentration as the integral of its
	concentration profile divided by the channel width.

\end{description}         


\myparagraph{Running time.}
The algorithm for profile computations takes only constant time to update the profile for each node and
channel, thus the overall running time is linear with respect to the 
size of the grid design. (Thus never worse than $O(mn)$ for an  $m\times n$ grid.)      
The overall running time is dominated by solving the linear system in part~(2).  
For grid sizes that might be of use in grid libraries, say up to $20\times 20$, 
Gaussian elimination is sufficiently fast.
For larger grids, one can take advantage of the sparsity of
the linear system to speed up the computation.


\section{Concentration Profile Model}
\label{sec: concentration profile model}


\begin{figure}
	\hfill
	\includegraphics[height = 1.6in]{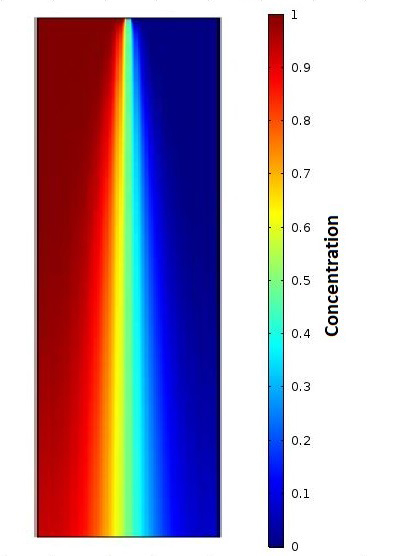}
	\hfill
	\includegraphics[height = 1.6in]{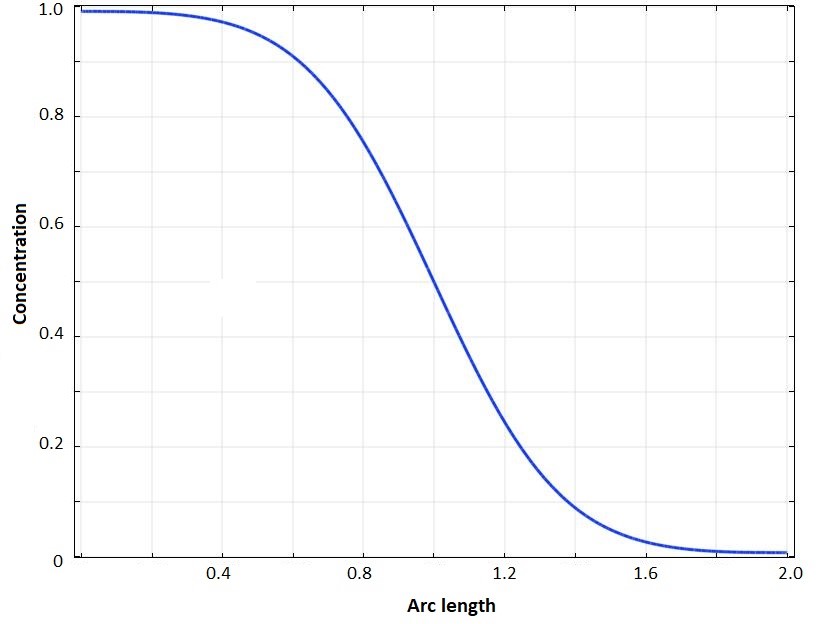}
	\hfill    {\ }
	
	\caption{The picture on the left shows the result of mixing of reactant (red) and
			buffer (blue) due to diffusion over some period of time.
			The figure on the right is a graph representing the concentration profile of this mixture
			at the end of the channel.
			} 
	\label{fig_profile}	
\end{figure}			


\myparagraph{True concentration profiles.}
Consider a mixture of two fluids, one reactant (with concentration 1) and the other buffer (with concentration 0), 
flowing along a straight channel
of some width $w$ and length $l$. 
For any distance $l'\le l$ from the beginning of this channel, 
a \emph{concentration} profile at $l'$
is a function that gives concentrations of all points in the channel along the line
segment (of length $w$) that is perpendicular to the channel and directed counter-clockwise to the flow.
We will be interested in how this profile evolves with time, that is with the value of $l'$ increasing.

Figure~\ref{fig_profile} shows an example. 
Reactant and buffer are injected at the same rate into the top opening of a vertical straight 
channel, and allowed to mix while they flow. 
Initially the two fluids are separated, with the reactant to the left of the buffer, so
the profile will be a 1/0 function.
The flow is assumed to be laminar and the two fluids will gradually mix as a result of diffusion. 
After a period of time, this mixing produces a non-uniform 
concentration profile shown in Figure~\ref{fig_profile} on the right. 
This concentration profile is a smooth curve with the leftmost 
region having concentration 1,  
the rightmost region having concentration 0, and the middle region contains partially mixed fluids with 
concentration decreasing from left to right. 
Using the diffusion model, this profile function can be determined
from the mixing time, which is the time it takes for the fluid to flow through the channel. 
Half of the width of this middle region is referred to as diffusion length and
denoted $L$ (normal to the flow direction, units $m$). It 
can be computed from the formula (see~\cite{kirby_book_micro_2013}):
\begin{equation}
	L = 2\sqrt{Dt}.
	\label{eqn: diffusion length formula}
\end{equation}
where $t$ is the mixing time and $D$ is the diffusion coefficient of the fluid (units $m^2/s$).

In microfluidic grids, the flow may be repeatedly split or different flows may get combined, and
the resulting concentration profile will not have the form in Figure~\ref{fig_profile} anymore;
in fact, profile functions that arise in such grids are too complex to be captured analytically.
Below we prove, however, that these profiles have a certain monotonicity
property that will allow us to approximate them by a simpler function.
Interestingly,
this monotonicity property involves the concept of the partial order that is dual to the flow pattern.


\myparagraph{Concentration monotonicity.}
The intuition behind the monotonicity property is illustrated in Figure~\ref{fig_profile}; intuitively, the
profile function in a channel should be monotonely decreasing from left to right. This property is trivial
if we start with a 1/0 profile (with pure reactant to the left of pure buffer) at the top of a vertical channel
and allow the fluid to diffuse when it flows down along the channel. 
However, in a mixing grid the flow pattern may be quite complex. 
For example, depending on a grid's structure and flow velocities at its inlets, the
flow direction in a vertical channel could be down or up.
As a result, even the notions of ``left'' and ``right'' are not well defined anymore. 
Joins and splits complicate this issue even more. 
Thus, to define this monotonicity property, we need to capture the notion of ``left-to-right direction''
not only with respect to one channel, but also between different channels. This notion will be
formalized using a partial ordering of the grid's channels. This partial order is defined as
the dual order of the flow pattern in the grid. Below we formalize these concepts.

\begin{figure}
\begin{minipage}[t]{0.4\linewidth}
	\includegraphics[height=2in]{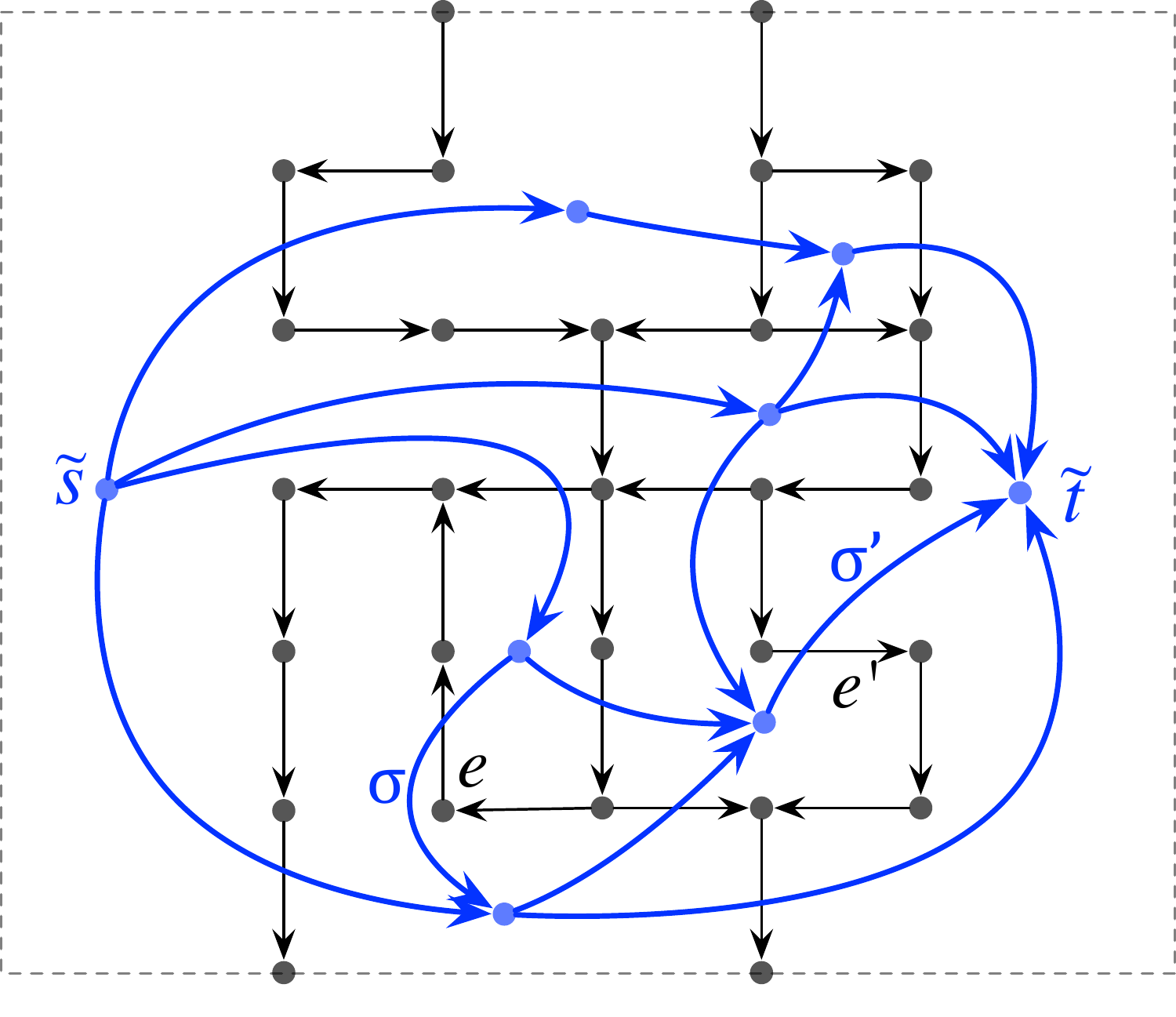}
	\caption{Grid representation graph $G$ (black) and its dual graph $\tildeG$ (blue).}
	\label{fig_dual}		
\end{minipage}
\hfill
\begin{minipage}[t]{0.4\linewidth}	
	\includegraphics[height = 1.6in]{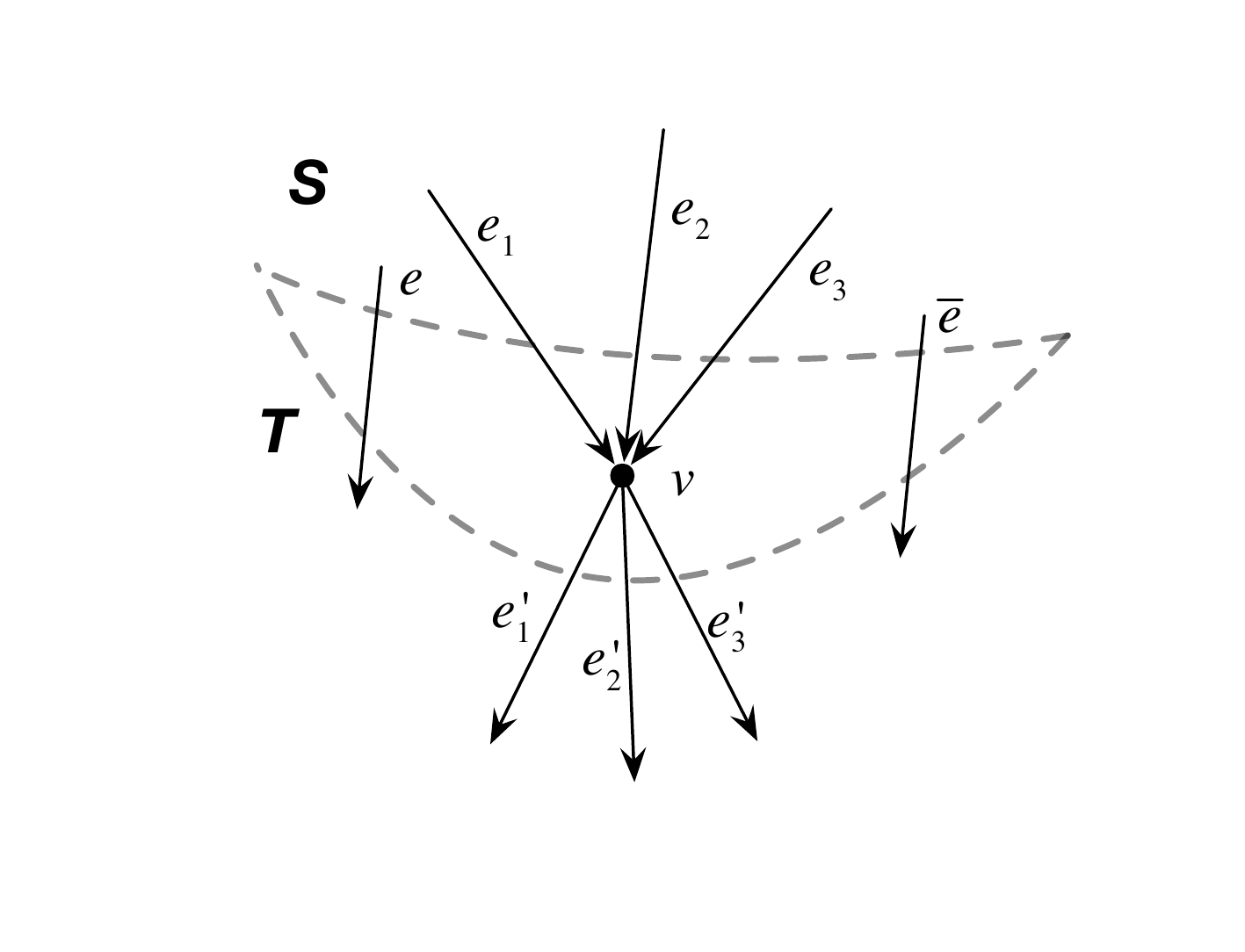}
	\caption{Moving a vertex $v$ from $T$ to $S$.} 
\label{fig_proof}	
\end{minipage}
\end{figure}			

Once the flow directions are computed, the flow pattern through the grid design 
can be naturally represented as a straight-line planar drawing of a DAG
(directed acyclic graph) whose nodes are grid points (including inlets and outlets) 
and edges are channels with directions determined by the flow. We will denote this graph by $G$.

Next, we construct a dual DAG $\tildeG$. To this end, enclose the grid 
in an axis-parallel rectangle, slightly wider than the grid,
with the inlets on its top edge and outlets on its bottom edge, as in Figure~\ref{fig_dual}. 
The grid (that is, the embedding of $G$) partitions this rectangle into regions.
For each region of $G$, we create a vertex of $\tildeG$. Two vertices $\phi$, $\psi$ of $\tildeG$
are connected by an edge if the boundaries of their corresponding regions share at least one edge of $G$.
The direction of the edge between $\phi$ and $\psi$ is determined 
as follows: pick any edge $(u,v)$ of $G$ shared by the regions of $\phi$ and $\psi$. This edge $(u,v)$ must have
a different orientation in the boundaries of
$\phi$ and $\psi$ (clockwise in one and counter-clockwise in the other).
Then the edge between $\phi$ and $\psi$ is directed from the node where $(u,v)$
is clockwise to the node where it is counter-clockwise. 
This definition does not depend on the choice of $(u,v)$.
We will refer to this edge as being dual to $(u,v)$. It can be shown
that $\tildeG$ is a DAG with a unique source $\tildes$ and unique sink $\tildet$ that
correspond to the regions on the left and right of the grid design, respectively.
(Except for secondary technical differences, the construction of such a dual can be found, for example,
in~\cite{tamassia_algorithms_1986,dibattista_tamassia_algorithms_1988}.)

We now use $\tildeG$ to define a partial order on the edges of $G$ (that is, the channels of our grid design).
Call two edges $e$, $e'$ of $G$ \emph{related in $G$} if they are both on the same inlet-to-outlet path;
otherwise call them \emph{unrelated in $G$}.
If $e$ and $e'$ are unrelated in $G$ then, denoting by $\sigma$ and $\sigma'$ their dual edges,
there must exist a path from $\tildes$ to $\tildet$ in $\tildeG$ that contains $\sigma$ and $\sigma'$.
If $\sigma$ is before $\sigma'$ on this path, then we write $e \dualprec e'$, and say that $e$ \emph{dually precedes} $e'$. 
Figure~\ref{fig_dual} shows an example.
It is not difficult to verify that relation ``$\dualprec$'' is a partial order on $G$'s edges.  

\begin{theorem}\label{thm: profile monotonicity}
(Concentration profile monotonicity.)
Consider a grid design as described in Section~\ref{sec: statement of the problem} (in particular,
the inlet concentrations are non-increasing) with some flow, and its corresponding graph $G$.
The concentration profiles in $G$ satisfy the following properties:
\begin{description} 
	\item{(cpm1)}	For any edge $e$ of $G$, each concentration profile across $e$ is a non-increasing function.
	\item{(cpm2)}	For any two edges $e$, $e'$ of $G$, if
					$e\dualprec e'$ then each concentration value in each profile across $e$ is at least as large as
					each concentration value in each profile across $e'$.
\end{description}
\end{theorem}

\begin{proof}
The formal proof of Theorem~\ref{thm: profile monotonicity} is omitted here due to lack of space; 
we only give a brief sketch. 
For any two edges $e,e'$ of $G$, we write $e \succcurlyeq e'$ to indicate that all concentration values 
in each profile across $e$ are as large as all concentration values in each profile across $e'$. 
With this definition, part (cpm2) says that $e\dualprec e'$ implies $e \succcurlyeq e'$.

Observe first that property (cpm1) is preserved as we move a profile along $e$ in the flow
direction, due to the properties of mixing (which shifts reactant mass from higher to lower values). With this
in mind, it is sufficient to select, arbitrarily, one ``representative'' profile for each edge and
to prove properties (cpm1) and (cpm2) only for these selected profiles.

Define a cut $(S,T)$ of $G$ in a natural way, as the partition of its vertices
such that each inlet-to-outlet path visits vertices in $S$ before visiting vertices in $T$.
We prove by induction on $|S|$ that all edges with tail vertex in $S$ satisfy the theorem.
This is true in the base case, when $S$ contains only inlets. In the inductive step,
the definition of cuts implies that there is a vertex $v$
that has all predecessors in $S$. Choose any such $v$. We move $v$ 
from $T$ to $S$ (see Figure~\ref{fig_proof}) and show that the inductive claim is preserved.
The clockwise ordering of $v$'s incoming edges $e_1,e_2,...$ is the same as their $\dualprec$ order. 
The same applies to the counter-clockwise order of its outgoing edges, $e'_1 \dualprec e'_2 \dualprec ...$.
We also know that the concentration profile at $v$ is the joined profile of $e_1, e_2, ...$ and then
it is split into decreasing concentration profiles of $e'_1, e'_2, ...$.
If there exist edges $e$ and $\bar{e}$ in $S$ that satisfy 
$e \dualprec e_1 \dualprec e_2 \dualprec ...\dualprec \bar{e}$,
then $e$ and $\bar{e}$ are also a predecessor and successor, respectively,
of $e'_1, e'_2, ...$ in the $\dualprec$ ordering. 
These observations and the inductive assumption applied to $e$ and $\bar{e}$
imply that $e \succcurlyeq e'_1 \succcurlyeq e'_2 \succcurlyeq ... \succcurlyeq \bar{e}$.
\end{proof}

\myparagraph{Approximate concentration profiles.}  
In our technique, we approximate concentration profiles by simple 3-piece linear functions specified by 
four parameters $\aleft$, $\aright$, $\dleft$, and $\dright$, as shown in Figure~\ref{fig_function}. 
In interval $[0,\dleft]$ the concentration is $\aleft$, 
in interval $[w-\dright,w]$ the concentration is $\aright$,
and the concentration linearly decreases in interval $[\dleft,w-\dright]$, from $\aleft$ to $\aright$.
Throughout the paper we will refer to such simplified profile functions as \emph{{\TPL}-functions}.
We have $w-\dleft-\dright = 2L$, where $L$ is the diffusion length value. 
The area under the profile curve, divided by the channel width $w$, represents the 
overall (average) concentration value of the fluid in this channel.

We remark that Theorem~\ref{thm: profile monotonicity} remains valid for our approximate
concentration profiles, instead of real ones. This is because its proof remains correct for
any (not necessarily physically valid) mixing process that preserves the mass of
reactant and buffer, and moves reactant mass rightward over time. Our approximate profile model
has this property.

In Section~\ref{sec: algorithm for concentration profile} we explain how we
can use {\TPL}-functions to compute approximate concentration profiles for each part of the grid design.   
\begin{figure}
\begin{minipage}[t]{0.4\linewidth}
	\includegraphics[height = 1.5in]{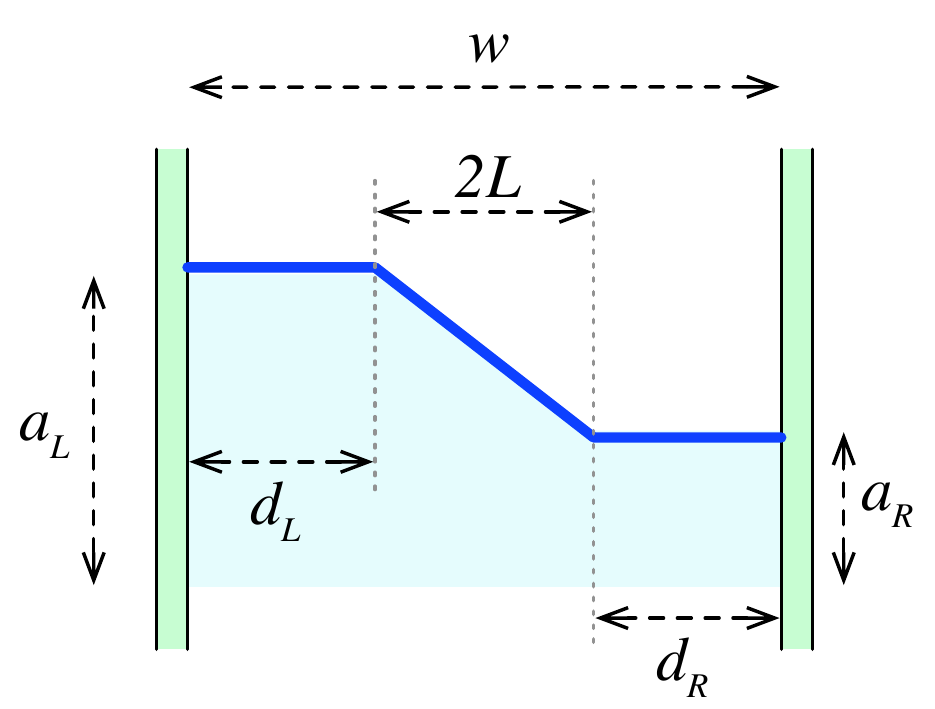}
	\caption{An {\TPL}-function profile.} 
	\label{fig_function}			
\end{minipage}
\hfill
\begin{minipage}[t]{0.4\linewidth}	
	\includegraphics[height = 1.5in]{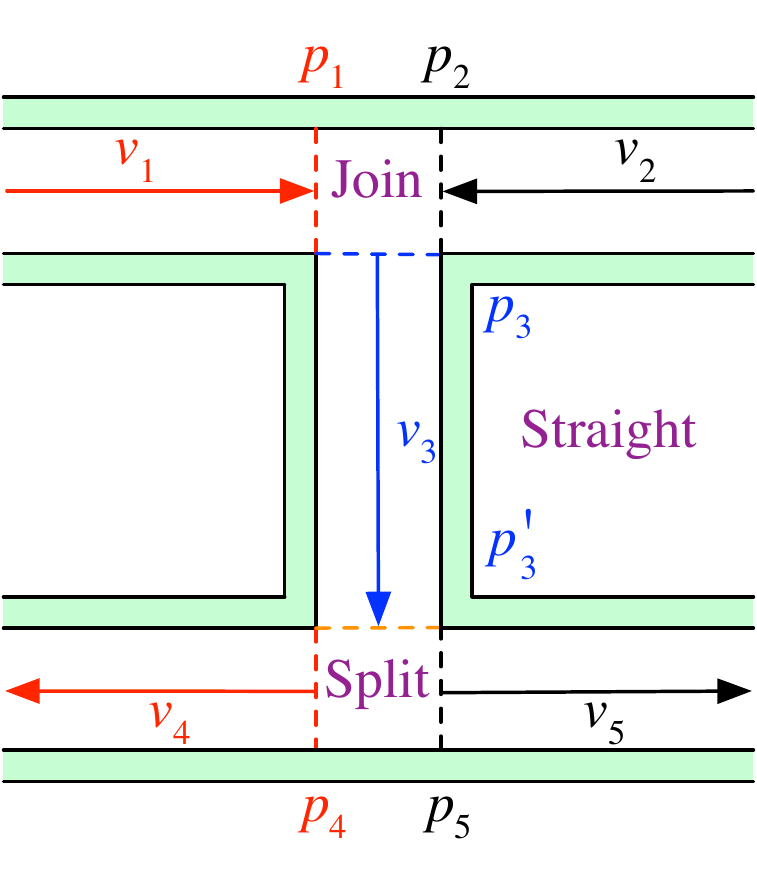}
	\caption{Grid partition: a join node on the top, 
	 followed by a downward straight channel and then a split node.} 
\label{fig_grid_partition}	
\end{minipage}
\end{figure}			


\section{Computing Flow Rate and Pressure}
\label{sec: computing flow}


The computation of the flow direction, flow rate and pressure in each channel of the grid is
quite straightforward, and it can be achieved by solving a system of linear equations. 
The unknowns are pressure values at every grid node and flow velocities in every channel. 
The first set of equations in this system are flow conservation equations: at every grid node
the total inflow is equal to the total outflow. The second set of equations are 
Hagen-Poiseuille equations which give relationship between
the flow rate and the pressure drop between the two ends of a channel:
$\Delta P = QR$,
where $Q$ is the volumetric flow, and $R$ is the flow resistance. As the
resistance value is the same in every channel segment and the input data specifies the velocity at 
the grid inlets, the exact value for $R$ is not needed and can be assumed to be $1$.

One thing to note is that in our setting we assume there is no friction at the walls of the channel,
thus implying that the velocity is uniform across the channel, which simplifies the formulas for
updating the profiles. (With friction, the velocity profile is a parabolic function.)
We found, however, that this assumption has only a negligible effect on the computed
concentration values.


\section{Algorithm for Estimating Concentration Profiles}
\label{sec: algorithm for concentration profile}

The core of our algorithm is a procedure for updating approximate concentration profiles
(represented by {\TPL}-functions) along the grid,
namely part~(4) of the overall algorithm in Section~\ref{sec: overview of algorithm}.
We describe this procedure in this section.

As mentioned earlier in Section~\ref{sec: overview of algorithm}, the grid is partitioned
into parts: straight channels and nodes, where nodes can be of several types depending on the
flow directions of its channels:
join nodes (2-way or 3-way), split nodes (2-way or 3-way), and combined join/split nodes 
(2 inflows and 2 outflows). 
This partitioning is illustrated in Figure~\ref{fig_grid_partition}.
In this figure, the channels are cut at points $p_1$, $p_2$, $p_3$, $p'_3$, $p_4$ and $p_5$,
producing one join node, one straight channel, and one split node.

In the join node at the top, flows 1 and 2 are joined into flow 3. 
The concentration profile at point $p_3$
is computed by combining {\TPL}-functions at points $p_1$ and $p_2$, taking the velocity ratios into account.
The combined function may not be an {\TPL}-function, and if so, we approximate it by
an {\TPL}-function.   

The straight channel stretches from the beginning of the channel at point $p_3$ to
its end point $p'_3$. The approximate {\TPL}-function at point
$p'_3$ is computed from the {\TPL}-function at point $p_3$ using the diffusion formula. 
(To account for mixing in the grid nodes, instead of the original channel length $l$, 
the algorithm uses a slightly larger value $\tilde{l} = l+w/2$, with $w$ is the width of the channel.)    

In the split node at the bottom, we use the {\TPL}-function at point $p'_3$ and the flow
velocities to compute the split {\TPL}-functions at points $p_4$ and $p_5$. This case is relatively
simple, as splitting an {\TPL}-function profile produces {\TPL}-functions, so no additional
simplification is needed here.


\subsection{Straight Channels}
\label{subsec: straight channels}

For a straight channel, the concentration profile changes over time with diffusion length increasing, that is with the middle 
portion of the {\TPL}-function becoming gradually longer and flatter. The {\TPL}-function may also
change its form, with $\dleft$ becoming $0$ or $\dright$ becoming $0$. We divide the evolution of the {\TPL}-function
along a straight channel into time intervals short enough so that in each such time
interval the profile function has the same form.

We now demonstrate how the {\TPL}-function changes in a straight channel during one of these time intervals. 
The fundamental idea behind our approach is this:  
Suppose that $\dleft > 0$ and $\dright > 0$ in Figure~\ref{fig_straight1} (Case~1 below). 
We then think of the current profile as
being a result of mixing process in a straight channel (or for stationary fluid)
that started from the original $1/0$-valued profile 
and lasted for some unknown time $t$.
This allows us to compute $t$ from the diffusion length formula~(\ref{eqn: diffusion length formula}). 
Once we have $t$, we can compute the new {\TPL}-function at a time
$t'=t+\Delta t$, where $\Delta t$ is chosen to be the time until the next ``event'', which is
the time when either the end of the channel is reached or when the profile form changes 
(that is, $\dleft$ or $\dright$ becomes $0$).
This leads to a relatively simple solution for updating
concentration profiles that satisfy $\dleft > 0$ and $\dright > 0$,
that is for {\TPL}-functions that have three non-empty segments.  
If $\dleft = 0$ or $\dright = w$, or both, a slightly less direct approach is required (Cases~2, 3 and 4 below).
\begin{figure}
	\hfill
	\includegraphics[height = 1.2in]{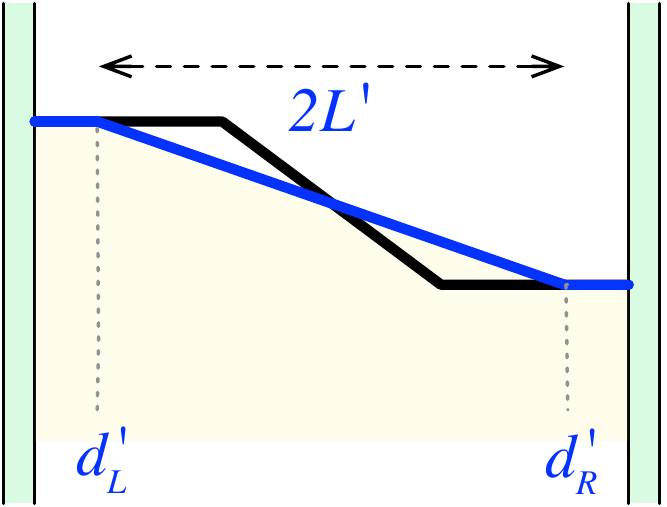}
	\hfill
	\includegraphics[height = 1.2in]{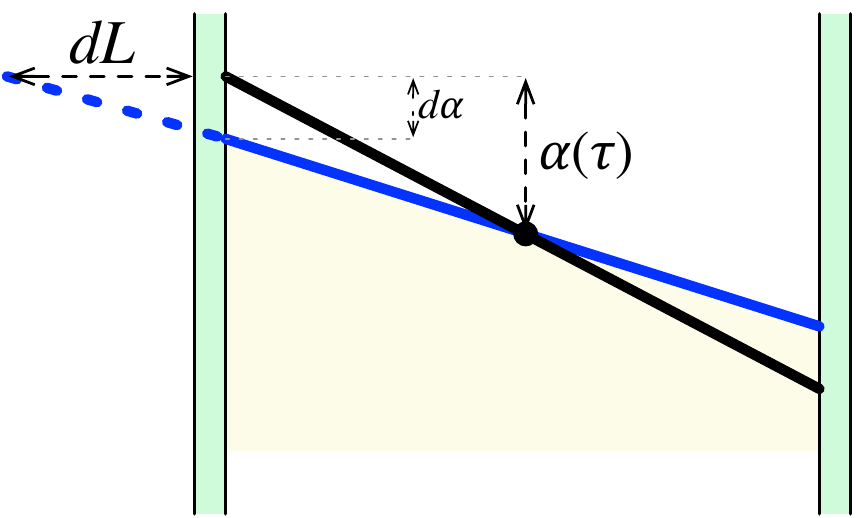}
	\hfill
	\caption{Updating the concentration profile in a straight channel in
		Case~1 (on the left) and Case~2 (on the right). 
		The ``before'' profile is black and the ``after'' profile is blue.} 
	\label{fig_straight 1 and 2}
\end{figure}	


\medskip
\noindent
\mycase{1} $\dleft > 0$ and $\dright > 0$ (see Figure~\ref{fig_straight 1 and 2}).	
		As explained above, 
		using formula~(\ref{eqn: diffusion length formula}) we compute $t= {L^2}/{4D}$.
		Next, we compute the new diffusion length $L'$ after time $\Delta t$:
		$L'= 2\sqrt{D(t+\Delta t)} =\sqrt{{L^2}+4D\Delta t}$.
		This will give us the new {\TPL}-function at time $t'$,
		defined by parameters
		$\dleft' = \dleft - (L'-L)$ and $\dright' = \dright - (L'-L)$. 
 		We choose $\Delta t$ to be the time until either the end of the channel is reached,
 		or until one of $\dleft$, $\dright$ becomes $0$, whichever happens first. 		


\medskip
\noindent
\mycase{2}	$\dleft = 0$ and $\dright = 0$ (see~Figure~\ref{fig_straight 1 and 2}).
		The channel boundaries introduce some distortion into the diffusion process that is difficult to model.
		Our approach here is based on the observation (verified experimentally) that this distortion has
		negligible effect on the overall profile. 
		Thus, we think about the current concentration profile as the linearly decreasing part of the
		profile 
		in a ``virtual'' wider channel in which the concentration profile has the same form as in Case~1,
		with $L = w/2 = 2\sqrt{Dt}$. The goal is to estimate the change of $\aright$ and $\aleft$ in time
		interval $[t,t']$.
		        
		For $\tau\in[t,t']$, define $\aleft(\tau)$ and $\aright(\tau)$ to be the
		concentrations at the left and right wall at time $\tau$, and 
		$\alpha(\tau) = (\aleft (\tau)- \aright(\tau))/2$. 
		We want to estimate the derivative $d\alpha/d\tau$.
		
		To this end, let $d L$ and $d\alpha$ denote the changes of $L$ and $\alpha(\cdot)$ 
		in a time interval $[\tau,\tau+d \tau]$, where $d\tau>0$ is very small. (In the calculations below
		we will approximate some values assuming that $d\tau$ approaches $0$.) 
		We have  $d L = 2\sqrt{D(\tau+d \tau)}-2\sqrt{D\tau} 
		= 2\sqrt{D}\,d\tau / (\sqrt{\tau+d \tau}+\sqrt{\tau}) 
		\approx d \tau \, \sqrt{D / \tau}$,  and, by simple geometry,  
		$d\alpha = -\alpha(\tau) \,d L/(L+d L)$.
	    	Substituting, we get  
		$d\alpha \approx - d\tau\, \alpha(\tau)/ 2\tau$.	
		Thus the derivative is $d \alpha(\tau) / d\tau = -\alpha(\tau) / (2\tau)$.

		Solving this differential equation and using the initial condition at $t$ gives us 
		$\alpha(\tau) = \alpha(t) \sqrt{t/\tau}$,  where $\alpha(t) = (\aleft - \aright)/2$. 
		Thus, at time $t'$ we have 
		$\alpha(t') = \alpha(t) \sqrt{t/t'}$. 
		From this we can calculate the new concentration values at
		the channel walls: $\aleft' = \aleft - \alpha(t)(1- \sqrt{t/t'})$ 
		and $\aright' = \aright + \alpha(t)(1- \sqrt{t/t'})$. 
		Here, $t'$ is taken to be the time when the end of the channel is reached.


\medskip
\noindent
\mycase{3}	 $\dleft > 0$ and $\dright = 0$  (see Figure~\ref{fig_straight 3 and 4}). 
   		The new concentration profile in this case is computed based on the assumption 
		that the value of $\dleft$ will decrease according to formula~(\ref{eqn: diffusion length formula}). 
		This gives us the new value $\dleft'$ for the new {\TPL}-function, namely
		$\dleft' = \dleft - (L'-L)$ for $L' =  2\sqrt{D(t+\Delta t)}  =  \sqrt{4D\Delta t + (w-\dleft)^2}$.

		Then we use molecular preservation property (that is the area under
		the concentration profile does not change) and straightforward calculation
		to obtain the new concentration at the right wall, $\aright'=\aleft - L(\aleft-\aright)/L'$.
		Here, $\Delta t$ is taken to be the time until either the end of the channel is reached or
		until $\dleft$ becomes $0$, whichever happens first.


\medskip
\noindent
\mycase{4}	 $\dleft = 0$ and $\dright >0$ (see Figure~\ref{fig_straight 3 and 4}).
	 	 This case is symmetric to Case~3, and similar calculation gives us
		 the values of  $\dright'$ and $\aleft'$.

\begin{figure}
	\hfill
	\includegraphics[height = 1.2in]{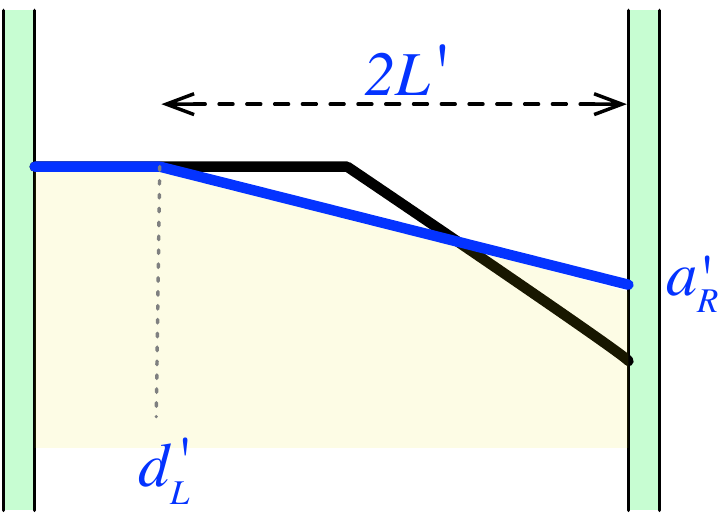}
	\hfill
	\includegraphics[height = 1.2in]{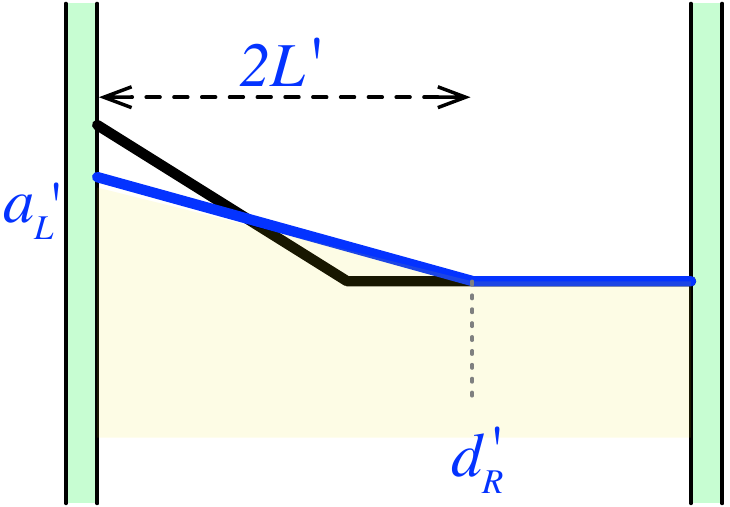}
	\hfill    {\ }
	
	\caption{Updating the concentration profile in a straight channel in
		Case~3 (on the left) and Case~4 (on the right).
		The ``before'' profile is black and the ``after'' profile is blue.} 
	\label{fig_straight 3 and 4}			
\end{figure}			

\subsection{Joining Concentration Profiles.}
\label{subsec: joining concentration profiles} 

To join flows, we first simply combine together the {\TPL}-functions of two or three joined channels, with
the portion of the channel width that each joined flow occupies being proportional to the velocities of the inflows,
as shown in Figure~\ref{fig_join1}. We will refer to this profile as the \emph{combined profile}.
In general, the combined profile will not be an {\TPL}-function.  
If this is the case, we will need to simplify it to an {\TPL}-function. This will be done in two steps.
First we will convert the combined profile into a \emph{tentative} {\TPL}-function 
(see Figure~\ref{fig_join2}),
which will be later adjusted to satisfy the reactant volume preservation property.

The correctness of joining profiles depends critically on Theorem~\ref{thm: profile monotonicity}. This theorem
implies that the concentration of the combined profile is non-increasing from left to right,
as illustrated in Figure~\ref{fig_join1}. This figure shows two profiles being combined. The case
of combining three profiles is essentially the same, as 
the middle channel's {\TPL}-function is only needed to compute the area 
under the combined function; its parameters can be ignored. Thus, to simplify the
description below, we will assume that we are dealing with a 2-way join node.

Let $\aleftleft$ and $\dleftleft$ represent the parameters of the combined profile inherited from
the corresponding parameters of the left joined channel. (So $\aleftleft$ is the concentration along the
left wall in the left channel, and $\dleftleft$ is the length of its {\TPL}-function's left flat segment,
but rescaled according to the channel's flow velocity.)
By $\arightright$ and $\drightright$ we denote the corresponding parameters inherited from the right joined channel.
We take these four values as the parameters of our \emph{tentative} {\TPL}-function profile.
This function is only tentative, because the area under this {\TPL}-function may be different than
that under the combined profile.

The final {\TPL}-function, whose parameters will be denoted $\aleft'$, $\dleft'$, $\aright'$ and $\dright'$, 
is computed from the tentative {\TPL}-function by adjusting its parameters to make
sure that the reactant volume (the area under the profile) is preserved. This is done as follows.
Let $A$ be the area under the original combined profile (before converting it into the tentative form). 
If the area under the tentative profile is smaller than $A$, then the left segment of the
final profile is the same as in the tentative profile, that is
$\aleft' = \aleftleft$ and $\dleft'=\dleftleft$, while the middle and right segments are adjusted to 
increase the area to $A$: We start with $\dright' = \drightright$ and $\aright' = \arightright$.
We then decrease $\dright'$ until either the area equals $A$ or $\dright'$ is reduced to $0$. If
$\dright'$ becomes $0$, we then increase $\aright'$ until the area becomes $A$. 
The case when the area under the tentative profile is larger than A is symmetric; in this case
the right segment of the final profile is the same as in the tentative profile, and we gradually
lower the middle and left segment to reduce the area to $A$.
\begin{figure}
	\captionsetup[subfigure]{justification=centering}
	\begin{subfigure}[t]{.3\textwidth}
		\centering
		\includegraphics[height = 1in]{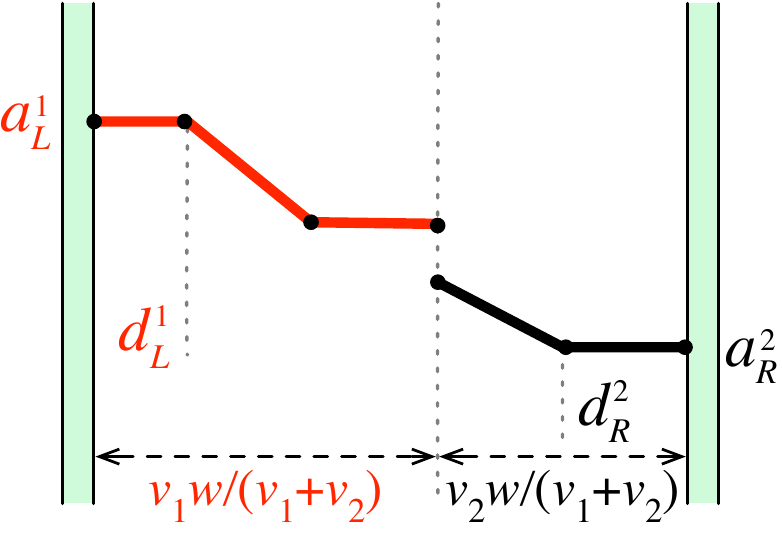}
		\caption{}
   		\label{fig_join1} 
	\end{subfigure}
	\hfill
	\begin{subfigure}[t]{.3\textwidth}
		\centering
		\includegraphics[height = 1in]{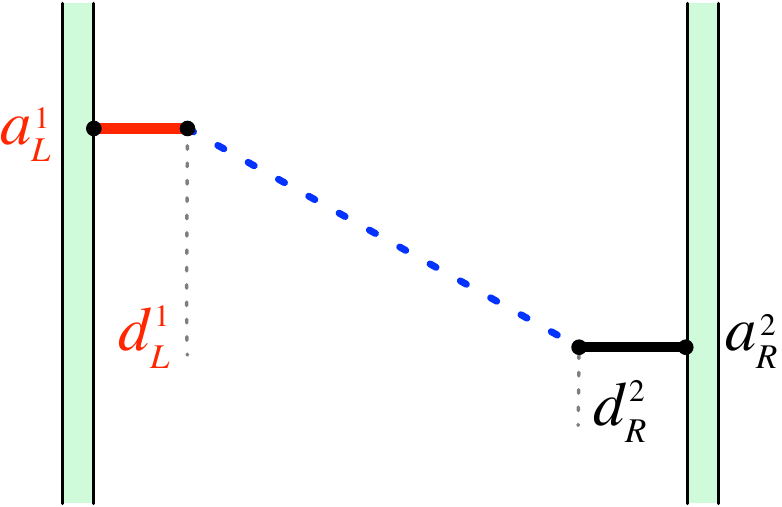}
		 \caption{}
   		\label{fig_join2} 
	\end{subfigure}
	\hfill
	\begin{subfigure}[t]{.3\textwidth}
		\centering
		\includegraphics[height = 1in]{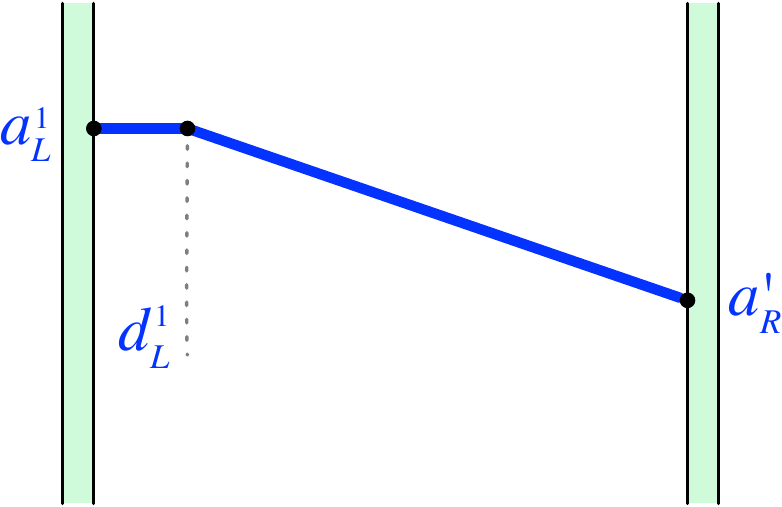}
		 \caption{}
   		\label{fig_join3} 
	\end{subfigure}	
	\caption{Joinning concentration profiles of two flows (red and black).
	Combined profile (a). The tentative {\TPL}-function (b).
	(c) one possible final {\TPL}-function in the case when the
	area under the tentative profile is too small.}
	\label{fig_join}
\end{figure}        


\subsection{Splitting Concentration Profiles}
\label{subsec: splitting concentration profiles}   

To split the flow at a node (either 2-way or 3-way),  the split profiles are determined by dividing the profile 
of the inflow proportionally to the velocity ratios of the outflows (see Figure~\ref{fig_split})
and rescaling appropriately. 
For example, for two-way splits, if the velocities of the left and right outflows are $v_1$ and $v_2$, then
the inflow profile is divided into two parts: the left one of width $v_1w/(v_1+v_2)$ and
the right one of width $v_2w/(v_1+v_2)$. 
This produces two profiles which are then rescaled to have width $w$. 
Conveniently, splitting a profile represented by a {\TPL}-function produces profiles that
are also in the same form, and the parameters of these new {\TPL}-functions
can be determined with a straightforward computation, by rescaling. 
Thus, in this case no further simplifications of the new profiles are needed.


\subsection{Join-and-Split Nodes.} 
 
In the grid there may also be nodes with two inflows and two outflows.
These nodes must have the form shown in Figure~\ref{fig_join_split}, namely the
inflow channels are adjacent, and so are the outflow channels.
(The other option, with the two inflows being opposite of each other,
is impossible, as this would imply the existence of a flow circulation in the grid.) 

We treat this case by reducing it to simple joins and splits, for which we
provided solutions earlier. This is done by breaking the computation of the new profile
into two sub-steps:
\begin{itemize}
	\item First, we join the two inflow {\TPL}-functions, following the method described in 
	Section~\ref{subsec: joining concentration profiles}.
 	\item Then, we take the computed {\TPL}-function profile and
		split it into two {\TPL}-functions for the
		outflow channels, following the method discussed
		 in Section~\ref{subsec: splitting concentration profiles}.
\end{itemize}

\begin{figure}      
\hfill
\begin{minipage}[t]{0.4\linewidth}
	\includegraphics[height = 1.2in]{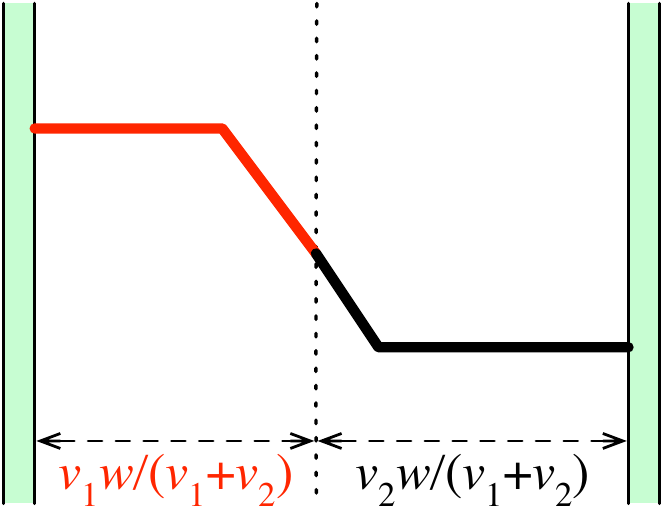}
	\caption{Splitting a profile into two profiles 
	(red and black).} 
	\label{fig_split}			
\end{minipage}
\hfill
\begin{minipage}[t]{0.4\linewidth}	
	\includegraphics[height = 1.2in]{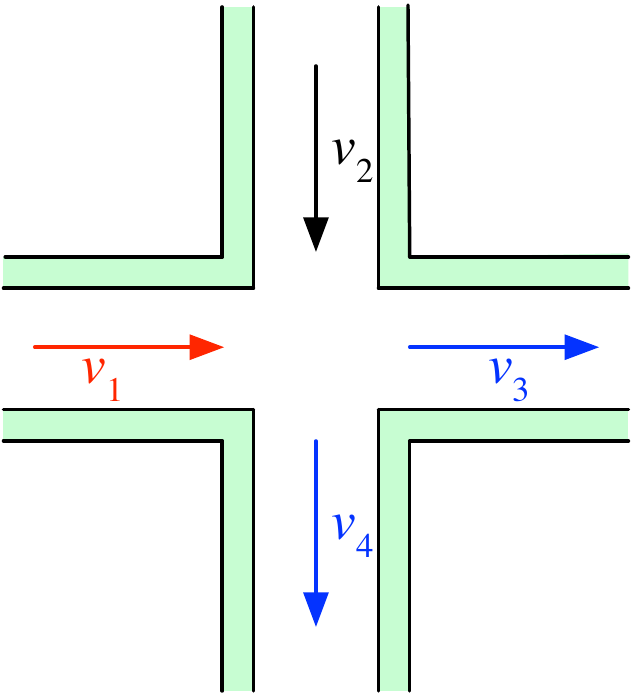}
	\caption{A join-and-split node. Flow $v_1$ and $v_2$ is joined and then split into $v_3$ and $v_4$.} 
	\label{fig_join_split}	
\end{minipage}
\hfill
\end{figure}			


\section{Experimental Results}
\label{sec: experimental results}


\begin{figure}[ht]
	\includegraphics[width=\linewidth]{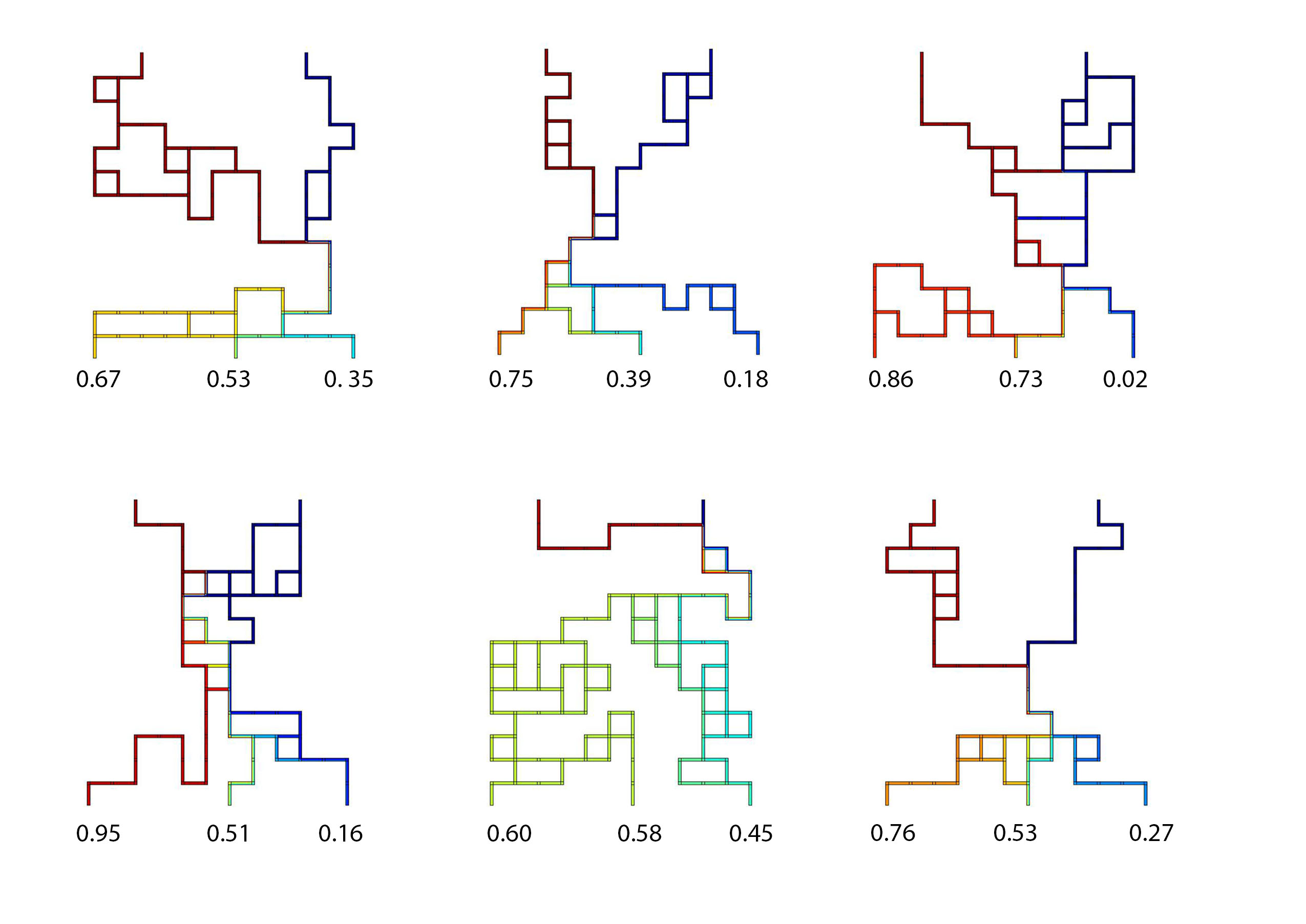}
	\caption{Randomly constructed $12\times 12$ grids with concentration values at the outlets.}
		\label{fig_sample}
\end{figure}

We used MATLAB® with COMSOL Multiphysics® via LiveLink™ \cite{comsol} to generate our test grids. 
Note that in the method from~\cite{wang_random_2016}, most generated chips have unreachable or
``dead-end'' channels, namely
channels that do not appear on any inlet-to-outlet path. Such channels are redundant, as they
have no effect on the mixing process. 
Our generator uses a similar approach as in~\cite{weiquing_more_2018}
to generate only grids that are connected and have no redundant channels.
Figure~\ref{fig_sample} shows four examples of random $12\times 12$ grids obtained from our generator
as well as their outlet concentrations. 

Our concentration prediction
algorithm is implemented in Python and tested on a 3.50GHz quad core 32GB RAM  workstation. 
We conducted an experimental comparison of our algorithm with COMSOL simulation. 
In our experiment, we used 200 $12\times 12$ sample random grids (obtained from our generator)
 with two inlets (the left has reactant with concentration 1 and the right has buffer with concentration 0)
 and three outlets. The flow velocities at the inlets are both of constant rate $1~mm/s$.
The reactant is sodium, with diffusion coefficient value $1.33\times10^{-4} mm^2/s$.
We used slower velocities than the chips of~\cite{wang_random_2016} 
in order to increase the time the fluid spends in the chip, which improves
variability of concentration values at the outlets.

For COMSOL simulations we used very fine triangular meshes containing 5-10 millions elements. 
This mesh size was determined through a mesh refinement study to ensure that this setting does not affect
outlet concentrations. For a random $12 \times 12$ grid we did a sweep on mesh element size
by reducing the size at each step and observed the changes in outlet concentrations, repeating until 
the changes are less than 1\% of maximum concentration.
\begin{figure}[h]
\begin{center}
\includegraphics[width=140mm]{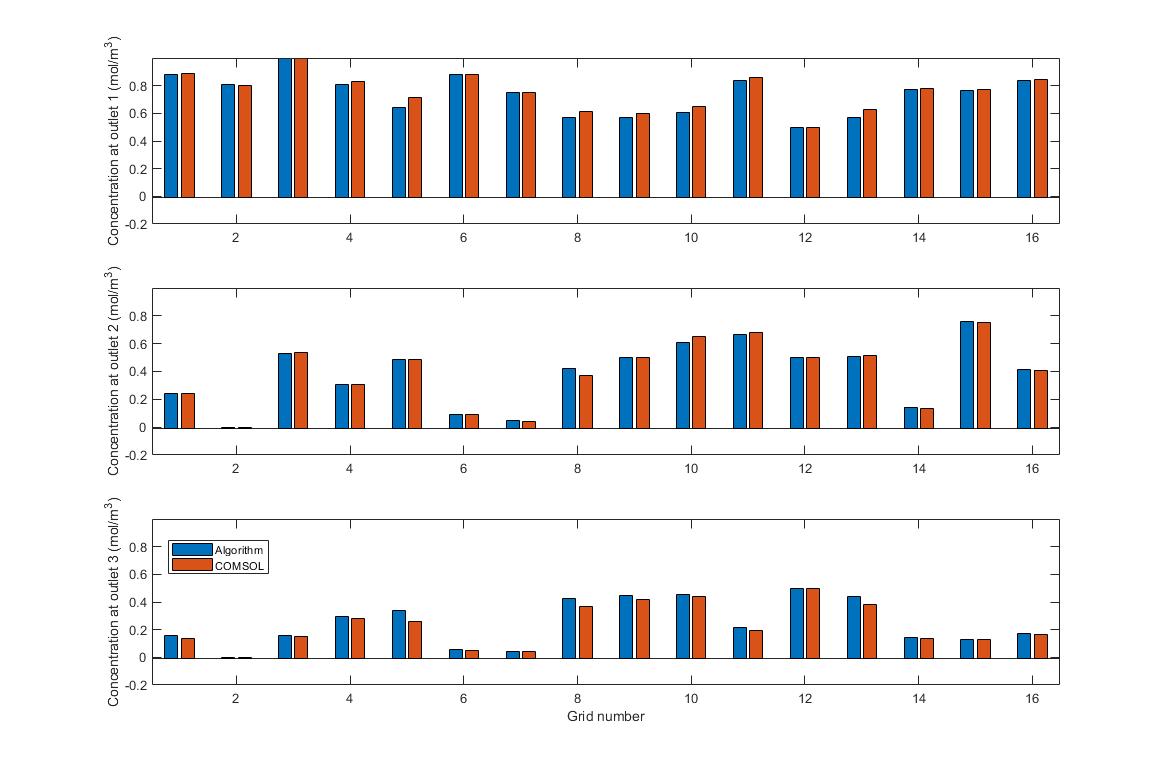}
\end{center}
\caption{Comparison of concentration values at the outlets for sixteen randomly selected grids. 
	Blue bars represent our algorithm and red bars represent COMSOL.}
\label{result_concentration}
\end{figure}

The results for fluid velocities at the 3 outlets are very consistent with COMSOL simulation, with 
the average percentage difference of velocity values being 0.8\%.
For concentration values at the outlets,
the average absolute difference is 0.006 $mol/m^3$, which is 0.6\% of the maximum concentration. 
The maximum absolute difference is 0.07 $mol/m^3$.
Figure~\ref{result_concentration} shows the difference of concentration values at the 3 outlets between 
the algorithm and COMSOL, on 16 randomly selected grid designs.

Execution times are measured on the same sample set of 200 grids. On average, 
with our mesh setting,
COMSOL takes approximately 21.3 minutes to finish the computation of one $12\times12$ grid. 
These times also vary significantly among different grids, with
the fastest time of about 6 minutes and the longest around 30 minutes. 
Our algorithm is several orders of magnitude faster, requiring on average only
0.0075 second to process one grid. It also uses much less memory: the memory used for concentration
profile computations is only linear in the side of the mixing grid design. 
COMSOL memory requirements are orders of magitude
higher, as it depends on the size of the mesh used for the simulation. 

We also performed an experiment to show that randomly generated $12\times 12$ grid designs can produce a
large collection of sufficiently different concentration vectors, and that  
a useful library of mixing grids can be successfully populated using our random grid generator and 
concentration prediction algorithm. 
We generated 50 million random designs on which we run our algorithm to compute output concentrations.
We then filtered out grids that were redundant, in the sense that their output concentration values differed
by no more than $0.01$. The resulting grid library consisted of 2600 different concentration vectors.


\section{Discussion}
\label{sec: discussion}



In this work we show that fluid mixing in microfluidic grids can be efficiently simulated
using simple 3-piece linear functions to model concentration profiles. Our algorithm 
is very fast, outperforming COMSOL simulations by several orders of magnitude in term of running time,
  while producing nearly identical results. 

While our current implementation of the algorithm assumes that the input is a grid design,
the overall technique applies to arbitrary acyclic planar graphs. The only required
assumptions involve inlets and outlets: all
inlets and outlets must be located on the external face, cannot interleave, and
the inlet concentrations must be non-decreasing in the clockwise direction along the external face.
This is not a significant restriction, as in a microfluidic device its inlets and outlets would
normally be located along its external boundary.

Another possible enhancement of our method is related to the assumption that the mixing process
in microfluidic grids is exclusively caused by diffusion. This is a valid assumption for flows
along straight channels, and it gives a good approximation of outlet concentrations 
for slow fluid velocity, typical for most microfluidic chips. 
Prior work in~\cite{elbow} showed, however, that with increased velocities, 
convection generated in channel bends affects mixing as well. 
(This will also likely apply to the split and join nodes.)
Taking such convections effects into account could further improve the accuracy of our method.


\bigskip
\myparagraph{Acknowledgements.}
We would like to thank P.~Brisk, W.~Grover, and V.~Rogers for enlightening us on
various aspects of microfluidics and fluid dynamics, 
and the anonymous reviewers for useful comments that helped us improve the 
presentation of this work.

\bibliographystyle{plain}

\bibliography{flows_in_grids_references}

\end{document}